\newif\iffull\fullfalse
\newif\ifapp\appfalse

\documentclass[sigplan,9pt]{acmart}

\acmConference[LICS'18]{Thirty-Third Annual ACM/IEEE Symposium on
Logic in Computer Science}{July 9--12,  2018}{Oxford}
\acmYear{2018}
\acmISBN{} 
\acmDOI{} 
\startPage{1}

\setcopyright{none}

\usepackage{booktabs}   
\usepackage{subcaption} 

\newif\ifwrong\wrongfalse
\newif\iftemp\tempfalse

\usepackage{mathtools}
\usepackage{wasysym}
\usepackage{enumerate}
\usepackage{stmaryrd}
\usepackage[mathscr]{euscript}
\usepackage{mathpartir}
\usepackage{graphicx}
\usepackage{etex}
\usepackage[all]{xy}

\newenvironment{pfsketch}{\begin{proof}
}{\end{proof}}

\usepackage{DSarrow} 
\input{commonmacros.sty}
\input{adrien.sty}
\input{cbv.sty}
\input{encodings.sty}
\input{equations.sty}
\input{davide.sty}
\input{preorders.sty}

\begin{document}

\newtheorem{remark}[theorem]{Remark}

\title{Eager Functions as Processes}
\author{Adrien Durier}
\affiliation{
  \institution{Univ. Lyon, ENS de Lyon, CNRS,\\ UCB Lyon 1, LIP UMR 5668}
}

\author{Daniel Hirschkoff}
\affiliation{
  \institution{Univ. Lyon, ENS de Lyon, CNRS,\\ UCB Lyon 1, LIP UMR 5668}
}

\author{Davide Sangiorgi}
\affiliation{
  \institution{Universit{\`a} di Bologna and INRIA}
}


\begin{abstract}
  We study Milner's encoding of the call-by-value $\lambda$-calculus
  into the $\pi$-calculus.  We show that, by tuning the encoding to two
  subcalculi of the $\pi$-calculus (Internal $\pi$ and Asynchronous Local
  $\pi$), the equivalence on $\lambda$-terms induced by the encoding
  coincides with Lassen's eager normal-form bisimilarity, extended to
  handle $\eta$-equality.  As behavioural equivalence in the
  $\pi$-calculus we consider contextual equivalence and barbed
  congruence. We also extend the results to preorders.

  A crucial technical ingredient in the proofs is the
  recently-intro\-du\-ced technique of unique solutions of equations,
  further developed in this paper.  In this respect, the paper also
  intends to be an extended case study on the applicability and
  expressiveness of the technique.




\end{abstract}



\copyrightyear{2018}
\acmYear{2018}
\setcopyright{licensedothergov}
\acmConference[LICS '18]{LICS '18: 33rd Annual ACM/IEEE Symposium on Logic in Computer Science}{July 9--12, 2018}{Oxford, United Kingdom}
\acmBooktitle{LICS '18: LICS '18: 33rd Annual ACM/IEEE Symposium on Logic in Computer Science, July 9--12, 2018, Oxford, United Kingdom}
\acmPrice{15.00}
\acmDOI{10.1145/3209108.3209152}
\acmISBN{978-1-4503-5583-4/18/07}

\maketitle

\section*{Introduction}
  Milner's work on functions as processes~\cite{milner:inria-00075405,encodingsmilner},
that shows  how the evaluation strategies of  {\em call-by-name 
  $\lambda$-calculus} and   {\em call-by-value
  $\lambda$-calculus}~\cite{Abr88,DBLP:journals/tcs/Plotkin75} can be faithfully mimicked in
the $\pi$-calculus,  
 is generally considered a
landmark in Concurrency Theory, and more generally in  Programming Language Theory.
 The comparison with the $\lambda$-calculus 
   is a significant  expressiveness test for
the $\pi$-calculus. 
More than that, 
it
promotes the $\pi$-calculus to be  a basis for general-purpose 
 programming languages in which communication is the fundamental
 computing primitive. 
From the $\lambda$-calculus point of view, the comparison  provides the means  to
study
 $\lambda$-terms
 in    contexts  other  than  purely  sequential ones,
 and with the  instruments available to reason about 
 processes.
Further, Milner's work, and the works that followed it, have contributed to
understanding and developing the theory of the $\pi$-calculus.

More precisely,   Milner shows the operational correspondence between
reductions in the $\lambda$-terms and in the encoding $\pi$-terms. 
He then uses the correspondence to prove that the encodings are
\emph{sound}, i.e., if the processes encoding 
 two $\lambda$-terms  are
behaviourally equivalent, then the source $\lambda$-terms are also
behaviourally equivalent in the $\lambda$-calculus.
Milner also shows that the converse,   \emph{completeness}, fails, intuitively because 
the  encodings allow one to test the $\lambda$-terms in all  contexts
of the $\pi$-calculus~--- more
diverse than those of the $\lambda$-calculus. 

The main problem that Milner work left open is the characterisation
of the equivalence  on $\lambda$-terms induced by the encoding, whereby
two $\lambda$-terms are equal if their  encodings are behaviourally
equivalent $\pi$-calculus terms. 
The question is largely independent of the precise form of 
 behavioural equivalence adopted in the $\pi$-calculus
because the encodings are deterministic (or at
least confluent). In the 
paper we consider contextual equivalence (that coincides with may
testing and trace equivalence) and barbed congruence (that coincides
with bisimilarity). 

For the call-by-name  $\lambda$-calculus, the answer was found shortly later
\cite{San93,cbn}: 
the equality induced is the equality of L{\'e}vy-Longo 
 Trees~\cite{LONGO1983153}, the lazy
variant of  B{\"o}hm Trees.  
It is actually also possible to 
obtain B{\"o}hm Trees, by modifying the call-by-name encoding so to allow also reductions
underneath a $\lambda$-abstraction, and by including divergence among the observables \cite{xian}. 
These results show that, at least for call-by-name, the $\pi$-calculus encoding, while not
fully abstract for the contextual equivalence of the $\lambda$-calculus, is in remarkable
agreement with the theory of the $\lambda$-calculus: several well-known models of the
$\lambda$-calculus yield   L{\'e}vy-Longo Trees
 or B{\"o}hm  Trees  as
their induced equivalence~\cite{levy75,LONGO1983153,barendregt1984lambda}.

For call-by-value, in contrast, the problem of identifying the equivalence induced by the
encoding  has remained open, for two
main reasons. First, tree structures in call-by-value are less studied
and less established  than in call-by-name. Secondly, proving
completeness of an encoding of $\lambda$ into $\pi$ requires sophisticated proof techniques. For
call-by-name, for instance, a central role is played by 
\emph{bisimulation up-to contexts}. For call-by-value, however,
existing proof techniques, including `up-to contexts',
 appeared not to be  powerful enough.  


In this paper we study the above open problem for call-by-value. 
Our main result is that the equivalence induced on $\lambda$-terms by their call-by-value 
encoding into the $\pi$-calculus  is 
\emph{eager normal-form bisimilarity} \cite{lassentrees,lassentrees2}. 
This is a tree structure for call-by-value,
proposed  by Lassen as the call-by-value  counterpart 
   of
L{\'e}vy-Longo Trees. 
Precisely we obtain the variant that is insensitive to
$\eta$-expan\-sion, called \emph{$\eta$-eager
normal-form bisimilarity}.

To obtain the results we have however to make a few adjustments
to Milner's encoding and/or specialise the target language of the encoding.
 These adjustments have to do with the presence 
 of free
outputs 
(outputs of known names) in the encoding.  
We show in the paper that 
this 
 brings problems when analysing $\lambda$-terms with free
variables: desirable call-by-value equalities  fail.
An example is given by the law:
\begin{equation}
  \label{eq:nonlaw}
  I(x\val)  = x\val
\end{equation}
where $I$ is $\abs zz$ and $\val$ is a value. 
\iffull
\else
Two possible
 solutions are:
\fi
\begin{enumerate}
\item rule out  the free outputs; this essentially means transplanting the encoding 
onto the Internal $\pi$-calculus~\cite{internalpi},  a version of the
  $\pi$-calculus in which any name emitted in an  output is fresh;

\item control the use of capabilities in the $\pi$-calculus; for
  instance taking Asynchronous Local
  $\pi$~\cite{localpi}  as the  target of the translation. 
 (Controlling capabilities allows one to impose a directionality on
names, which, under certain
technical conditions, may  hide the identity of the emitted names.)
\end{enumerate} 

In the paper we consider both approaches, and show that 
in both cases, the equivalence induced coincides with
$\eta$-eager
normal-form bisimilarity.

In summary, there are two contributions in the paper:
\begin{enumerate}
\item Showing that Milner's encoding fails to equate terms that should be equal in call-by-value.
\item Rectifying the encoding, by considering different target calculi, and  
investigating Milner's problem in such a setting.
\end{enumerate}
The rectification we make does not really change the essence of the encoding -- 
in one case, the encoding actually remains the same. Moreover, the languages used 
are well-known dialects of the \pc, studied in the literature for other reasons. 
In the encoding, they allow us to avoid certain accidental misuses of the names 
emitted in the communications. The calculi were not known at the time of Milner's 
paper \cite{encodingsmilner}.

A key role in the completeness proof is played by a technique of
\emph{unique solution of equations}, recently proposed~\cite{usol}. 
The structure induced by Milner's 
call-by-value encoding was expected to look like Lassen's trees; 
however existing proof 
techniques did not seem powerful enough to prove it. 
The unique solution technique allows one to derive process bisimilarities from
equations whose infinite unfolding does not introduce divergences, by
proving that the processes are solutions of the same equations.  The
technique can be generalised to possibly-infinite systems of
equations, and can be  strengthened by allowing certain
  kinds of divergences in equations.  In this respect, another goal
of the paper is to carry out an extended case study on the
applicability and expressiveness of the techniques.  Then, a
by-product of the study are a few further developments of the
technique. 
In particular, one such result allows us to transplant uniqueness of
solutions from a system of equations, for which divergences
are easy to analyse, to another one. Another result is about
the application of the technique to preorders.

Finally, we consider 
preorders~--- thus referring to the
preorder on $\lambda$-terms induced by a behavioural preorder on their $\pi$-calculus
encodings. 
We introduce a preorder on Lassen's trees (preorders had not been considered by Lassen)
and show that this is the preorder on $\lambda$-terms induced by the call-by-value
encoding, when the behavioural relation on $\pi$-calculus terms is 
the ordinary contextual preorder (again, with the caveat of points (1)
and (2) above). 
With the move from equivalences to preorders, the overall structure of
the proofs of our full abstraction results remains the same. 
However, the impact on the application of the unique-solution
technique is substantial, because the phrasing of this technique in
the cases of preorders and of  equivalences is quite different.




\paragraph{Further related work.}
The standard behavioural equivalence in the $\lambda$-calculus
is contextual equivalence. Encodings into the $\pi$-calculus
(be it for call-by-name or call-by-value) break  contextual equivalence
 because $\pi$-calculus contexts
are richer than those in the (pure) $\lambda$-calculus. In the paper
we try to understand how far beyond contextual equivalence the
discriminating power of the $\pi$-calculus brings us, for
call-by-value.   
The opposite approach is to restrict the set of 'legal' $\pi$-contexts
so to remain faithful to contextual equivalence. This approach has been
followed, for call-by-name, and using type systems, in
\cite{BHYseqpi,toninho:yoshida:esop18}. 

Open call-by-value  has been studied in~\cite{accattolicbv}, where the
focus is on operational properties of $\lambda$-terms; behavioural
equivalences are not considered.
An extensive presentation of call-by-value, including denotational
models, 
is  Ronchi della Rocca and
Paolini's book~\cite{DBLP:series/txtcs/RoccaP04}.

In~\cite{usol}, the unique-solution technique is used in the 
completeness proof for Milner's call-by-name encoding. That proof 
essentially revisits the proof of~\cite{cbn}, which is based on 
bisimulation up-to context. We have  explained
 above that the case for
call-by-value  is quite different.

\paragraph{Structure of the paper.} We recall basic definitions about
the call-by-value $\lambda$-calculus and the $\pi$-calculus in
Section~\ref{s:background}. The technique of unique solution of
equations is introduced in Section~\ref{s:usol}, together with some
new developments. Section~\ref{s:enc:cbv}
presents our analysis of Milner's encoding, beginning with the shortcomings
related to the presence of free outputs. 
{
 The first solution to these shortcomings is to move to the Internal
$\pi$-calculus: this is described in Section~\ref{s:enc:pii}. }
For the proof of completeness, in Section~\ref{s:complete},
we  rely on unique solution of equations; we also compare 
such technique with the `up-to techniques'.  
{
The second solution is to move to the Asynchronous Local
  $\pi$-calculus: this is discussed in 
Section~\ref{s:localpi}.}
We show in Section~\ref{s:contextual} how our
results can be adapted to   preorders and to contextual equivalence.
Finally in Section~\ref{s:concl}  we highlight  conclusions and
possible future work.


\section{Background material}
\label{s:background}

Throughout the paper,  $\R$ ranges over relations.
The composition of  two   relations
$\R$ and $\R'$  is written 
$\R \: \R'$.
We often use infix notation for relations; thus 
 $P \RR Q$ means ${(P, Q)}\in\R$. 
A tilde  represents  a tuple.  
The $i$-th element of a tuple $\til P$ is  referred to as $P_i$. 
 Our notations are extended to tuples componentwise. Thus  
 $\til P \RR \til Q$ means  $P_i \RR Q_i$
for all components.

\subsection{The call-by-value  \lc}
We let $x$ and $y$ range over the set of $\lambda$-calculus  variables.
The set   $\Lao$  of $\lambda$-terms is  
 defined by the  grammar  
\begin{center}
$ M := \;   x    \midd   \lambda   x. M  \midd  
  M_1 M_2 \, .$ 
\end{center}
 Free variables, closed terms, substitution, 
$\alpha$-conversion
etc.\ are  defined as usual  \cite{barendregt1984lambda,DBLP:books/cu/HindleyS86}.
Here and in the rest of the paper (including when reasoning about
$\pi$ processes), we adopt the usual
``Barendregt convention''. This will allow us to assume freshness
of bound variables and names  whenever needed.
The set of   free variables 
in the term $M$  is   $\fv M$.
We group brackets on the left; therefore $M N L $  
is $(M N ) L$.
 We  abbreviate $\lambda  x_1. \cdots. \lambda  x_n.M $   as 
$\lambda  x_1 \cdots x_n.M $, or $\lambda  \tilde{x}. M$ if the length of
$\tilde x$ is not important.
Symbol $\Omega  $ stands for the  always-divergent term
 $(\lambda  x . x x)(\lambda  x . x x)$.



\txthere{explain the following things ("defined as usual"?):}{scope
  and associativity in the \lc, substitutions, free variables
  $\fv{M}$, $\fv{M,N}$ is short for $\fv M\cup \fv N$, values and
  evaluation contexts as syntactic categories (of terms and contexts),
  open terms (here, everything is defined for open terms)} 


A \emph{context} is a term with a hole \hdot,
possibly occurring more than once. If $C$ is a context, $C[M]$ is
a shorthand for $C$ where the hole \hdot is substituted by $M$. An
 \emph{evaluation context} is a special kind of 
context,
with exactly one hole \hdot, and in which the inserted term can
immediately run. In the pure $\lambda$-calculus \emph{values} are abstractions and variables.
\begin{center}
\begin{tabular}{rl}
{Evaluation contexts} & $\evctxt\scdef \hdot\OR \evctxt M\OR\val \evctxt$ \\
{Values} & $\val\scdef x \OR \abs x M$
\end{tabular} 
\end{center}
In call-by-value, substitutions replace variables with values; we call
them \emph{value substitutions}. 



Eager reduction (or $\betav$-reduction), 
  ${\red}
\subseteq    \Lao \times  \Lao $, 
is determined  by the rule:
\txthere{}{Eager reduction relation (definition?):}
$$  \evctxt[(\lambda x . M) \val]\red \evctxt[M\{\val/x\}]
\enspace.
$$

We write \reds\ for the reflexive transitive closure of \red.
A term in \emph{eager normal form} is a term that has no eager
reduction.

\begin{proposition}
\begin{enumerate}
\item If $M\red M'$, then $\evctxt[M]\red \evctxt[M']$ and $M\sigma
  \red M'\sigma$, for any value substitution $\sigma$. 
\item Terms in eager normal form are either values or of the shape
  $\evctxt[x \val]$.
\end{enumerate}
\end{proposition}

Therefore, given a term $M$, either $M\reds M'$ where $M'$ is a term in
eager normal form, or there is an infinite reduction sequence 
starting from $M$.
In the first case,  $M$ \emph{has eager
normal form $M'$}, written $M\converges M'$, in the second $M$
\emph{diverges}, written $M\diverges$. We write $M\converges$ when
$M\converges M'$ for some $M'$.


\begin{definition}[Contextual equivalence]\label{d:ctxeq}
Given $M,~N\in \Lao$, we say that $M$ and $N$ are contextually
equivalent, written $M\ctxeq N$, if for any context $C$, we have 
$C[M]\converges$ iff $C[N]\converges$. 
\end{definition}

\subsection{Tree semantics for call-by-value}
\label{s:back:trees}
We recall 
\emph{eager normal-form bisimilarity} \cite{lassentrees,lassentrees2,lassenfa}.

\begin{definition}[Eager normal-form bisimulation] 
\label{enfbsim}
A relation $\R$ between \lterms is an \emph{eager normal-form bisimulation} if, whenever
$M\RR N$, one of the following holds: 
\begin{enumerate}
\item both $M$ and $N$ diverge;
\item 
\label{ie:split}
$M\converges \evctxt[x\val]$ and $N\converges \evctxtp[x\valp]$ for
some $x$, values $\val$, $\valp$, 
and evaluation contexts 
 $\evctxt$ and $\evctxtp$ with $\val\RR
\valp$ and $\evctxt[z]\RR \evctxtp[z]$ for a  fresh $z$;
\item $M\converges \abs x M'$ and $N\converges \abs x N'$ for some $x$, $M'$, $N'$ with $M'\RR N'$;
\item $M\converges x$ and $N\converges x$ for some $x$.
\end{enumerate}
\emph{Eager normal-form bisimilarity}, $\enf$, is the largest eager normal-form bisimulation.
\end{definition}


Essentially, the structure of a $\lambda$-term that is unveiled by
Definition~\ref{enfbsim} is that of a (possibly infinite) tree 
obtained by repeatedly applying $\betav$-reduction, and branching a tree whenever
instantiation of a variable is needed to continue the reduction (clause \reff{ie:split}).  
We call such trees \emph{Eager Trees} (ETs) and accordingly also call  
 eager normal-form bisimilarity the \emph{Eager-Tree equality}.


\begin{example}
\label{exa:cteq}
Relation $\enf$ is strictly finer than contextual equivalence
$\ctxeq$: the inclusion ${\enf} \subseteq {\ctxeq}$ follows from the
congruence properties of $\enf$ \cite{lassentrees}; for the strictness,
examples are the following equalities, that hold for $ \ctxeq$ but not
for $\enf$:
  $$
  \Omega  = (\abs y \Omega) (x\val)
\qquad
x\val =  (\abs y x\val)(x\val)
\enspace.
$$
\end{example} 
\begin{example}[$\eta$ rule]
\label{exa:eta}
The  $\eta$-rule is not valid for $\enf$. For instance, we have 
$\Omega \not\enf \abs x \Omega x$.
The rule is not even valid on values, as we also have
$ \abs y x y \not\enf x$. It holds however 
for abstractions: 
$   \abs y (\abs xM) y\enf \abs xM
$ when $y\notin\fv{M}$.
\end{example} 
The failure of the   $\eta$-rule $ \abs y x y \not\enf x$ is
troublesome as, under any closed value substitution, the two terms are
indeed {
eager normal-form bisimilar} (as well as contextually equivalent). 
%
%
%
%
Thus \emph{$\eta$-eager normal-form bisimilarity}~\cite{lassentrees} takes
$\eta$-expansion into account so to recover such missing equalities.


\begin{definition}[\enfbsim] \label{enfebsim}
A relation $\R$ between \lterms is an \emph{\enfbsim} if, whenever $M\RR N$, either one of the clauses of Definition \ref{enfbsim},
or one of the two following additional clauses,  hold:
\begin{enumerate}
\setcounter{enumi}{4}
\item\label{lab:five} $M\converges x$ and $N\converges \abs y N'$ for some
 $x$, $y$, and $N'$ such that $ N'\converges\evctxt[x\val]$,
with      $y\RR \val$ and $z\RR
  \evctxt[z]$ for some value $\val$, evaluation context 
\evctxt, and  fresh $z$.
\item\label{def:enfe:case:eta}
the converse of \reff{lab:five}, i.e., 
$N\converges x$ and
  $M\converges \abs y M'$ for some
  $x$, $y$, and $M'$ such that $ M'\converges\evctxt[x\val]$,
with
     $ \val\RR y$ and 
$
   \evctxt[z]\RR z$ for some value $\val$, evaluation context 
 \evctxt, and  fresh $z$.
\end{enumerate}
Then \emph{$\eta$-eager normal-form bisimilarity}, $\enfe$, is the largest $\eta$-eager normal-form bisimulation.
\end{definition}
We sometimes call relation $\enfe$ the \emph{$\eta$-Eager-Tree equality}. 
\begin{remark}Definition~\ref{enfebsim} coinductively allows
$\eta$-expansions to occur underneath other $\eta$-expansions, hence 
trees with infinite $\eta$-expansions may be equated with finite trees.
For instance, $$x\enfe \abs y xy\enfe \abs y x (\abs z yz)\enfe \abs y x (\abs z y(\abs w zw))\enfe\dots$$
A concrete example is given by taking a fixpoint $Y$, and setting 
$f \defi (\abs {zxy} x(zy))$. We then have  $Yfx \reds \abs y x (Yfy)$, and then 
$x(Yfy)\reds x(\abs z y(Yfz))$, and so on. Hence, 
we have  $x\enfe Yfx$.
\end{remark}


\subsection{The \pc, \Intp\ and \alpi}
\label{s:back:pi}


In all  encodings
we consider, 
  the encoding of a  $\lambda$-term
 is parametric on  a name, 
i.e., it is a 
 function from names  to $\pi$-calculus  
 processes. We also need parametric processes (over one or several names) for writing recursive process definitions
 and equations. 
We call such  parametric processes {\em abstractions}.
The actual instantiation of the parameters of an abstraction $F$ is done
via the {\em application} construct $\app F \tila$.
 We use $P,Q$ for processes, $F$ for abstractions.
Processes and abstractions  form the set  of  {\em $\pi$-agents} (or
simply \emph{agents}), ranged
over by $A$. 
Small letters 
 $a,b, \ldots, x,y, \ldots$  
range  over the infinite set of names.
The grammar of the $\pi$-calculus is thus:
$$ \begin{array}{ccll}
A & := & P \midd F & \mbox{(agents)}\\[\mypt]
P & := & \nil    \midd    \inp a \tilb . P    \midd    \out a \tilb . P 
   \midd    \res a P
& \mbox{(processes)}  \\[\myptSmall]
& &
   \midd     P_1 |  P_2   \midd  ! \inp a \tilb . P   
\midd \app F \tila
\\[\mypt]
F & := & \bind \tila P \midd K & \mbox{(abstractions)}
   \end{array}
 $$

In  prefixes $\inp a \tilb$ and $\out a \tilb$, we call 
 $a$ the {\em subject} and  $\tilb$  the {\em object}. 
When the tilde  is empty, the surrounding brackets in prefixes
 will be omitted.
We  
 often abbreviate 
$\res a \res b P$ as $\resb{
a,b} P$.
An input prefix  $a (\tilb) .P$, a restriction 
   $\res{b} P$, and an abstraction $\bind\tilb P$  are binders for 
   names  $\tilb$ and $b$, respectively, and  give
rise in the expected way  to the definition of {\em free names}
(\mbox{\rmsf fn}) and {\em
bound names} (\mbox{\rmsf bn})
 of a term or a prefix,    and   $\alpha$-conversion.
An agent is \emph{name-closed}  if it does not contain free names. 
As in the $\lambda$-calculus, following the 
usual Barendregt convention
we   identify
processes or actions which only  differ on the choice of the 
 bound names. 
The symbol $=$  will 
mean  ``syntactic identity modulo
$\alpha$-conversion''.
Sometimes, we  use $\defi$   as abbreviation mechanism, to
assign a name to an   expression to which we want to refer later.



   We use constants, ranged over by $K$ for writing recursive definitions. Each
constant has a defining equation of the form 
$K \Defi  \bind{\tilx} P$, where $\bind{\tilx} P  $ is name-closed; $\tilx$ 
 are the formal parameters of
the constant (replaced by the actual parameters whenever the constant
is used).

Since the calculus is polyadic, 
we assume a \emph{sorting system}~\cite{milner1993polyadic}
   to avoid disagreements  in the arities of
the tuples of names 
carried by a given name 
and in applications of abstractions.
We will not present the sorting system 
because it is not essential. 
The reader should
take for granted that all agents described  obey  a sorting. 
A \emph{context}  $\qct$ of $\pi$    is a $\pi$-agent in which some
subterms have been replaced by the hole $\hdot{}$ or, if the context is
polyadic, with indexed holes $\holei 1, \ldots, \holei n$; 
then 
 $\ct A$ or $\ct {\til A}$
 is the agent resulting from replacing the holes with the terms $A$ or
 $\til A$.

We omit the operators of sum and matching (not needed in
the encodings).
\iffull
\else
We refer to~\cite{milner1993polyadic} for detailed discussions on the
operators of the language. 
\fi
We assign parallel composition the lowest precedence among the
operators.

\paragraph{Operational semantics.}\ifapp
\else
The operational semantics of the $\pi$-calcu\-lus  is standard
\cite{SW01a} (including the labelled transition system).
\fi
The reference behavioural equivalence for $\pi$-calculi will be
the usual \emph{barbed congruence}.
We recall its definition,  on a generic subset $\LL$ of
$\pi$-calculus processes.
 A \emph{$\LL$-context} is a  process of $\LL$ with a single
hole $\contexthole$ in it (the hole has a  sort too, as it could be in
place of an abstraction). 
We write $P \Dwa_{a}\;$  
if $P$ can make an output action
whose subject is $a$, possibly after some internal moves. (We make only
output observable because this is standard in asynchronous
calculi; adding also observability of inputs does not affect barbed
congruence  on the synchronous calculi we will consider.) 

\begin{definition}[Barbed  congruence]
\label{d:bc}
 {\em Barbed bisimilarity} is the largest 
symmetric relation $ \wbb$ on 
 $\pi$-calculus  processes   
such that
  $P \wbb  Q$ implies:
\begin{enumerate}
\item 
If $P \Longrightarrow P'$ then  there is $ Q'$ such that
$Q \Longrightarrow Q'$ and
      $P' \wbb Q'$.
\item  $P \Dwa_{ a}\/$  iff $Q \Dwa_{a}\/$.
\end{enumerate}
Let $\LL$ be a set of $\pi$-calculus agents, and 
 $A, B \in \LL$. We say that $A$ and $B$ are 
  {\em barbed congruent in 
$\LL$}, written $A \wbc\LL B$,
if for each (well-sorted)  $\LL$-context $C$, it holds that $C[A] \wbb C[B]$.
\end{definition}

\begin{remark}
\label{r:abs} 
Barbed congruence has been uniformly defined on processes and
abstractions (via a quantification on all process contexts). 
Usually, however,  
 definitions will only be given for processes; it is
then intended that they are extended to abstractions by requiring
closure under ground parameters, i.e., by supplying fresh names as
arguments. 
\end{remark} 

As for
 all contextually-defined behavioural relations, so 
barbed congruence is
hard to work with.  In all calculi we consider, it can be
characterised in terms of \emph{ground bisimilarity}, under the (mild)
condition that the processes are image-finite up to $\bsim$. 
{
 (We recall  that
the class of processes {\em image-finite up to \bsim} 
is the largest subset ${\mathcal {IF}}$ of
$\pi$-calculus processes  which is derivation  closed  and such that 
$P \in {\mathcal {IF}}$ implies that,   for all actions $\mu$, 
the set 
$\{P'  \st    P \Arr{\mu } P'\}$
quotiented by   \bsim\
is finite. The definition is extended to abstractions as by
Remark~\ref{r:abs}.) }
All the agents in the paper, including those obtained by encodings
of the $\lambda$-calculus, are image-finite up to \bsim. 
The distinctive feature of \emph{ground} bisimilarity is that it does not
involve instantiation of the bound names of inputs (other than by
means of fresh names), and similarly for abstractions.
In the remainder, we omit the adjective `{ground}'. 

\begin{definition}[Bisimilarity]
\label{d:bisimulation}
A symmetric relation $\R$ 
on $\pi$-pro\-cesses is a
\emph{bisimulation}, if whenever $P \,\R\, Q$ and $P \arr\mu P'$, then $Q \Arcap\mu Q'$
for some $Q'$ with $P' \,\R\, Q'$. 

Processes $P$ and $Q$ are \emph{bisimilar}, written
$P\approx Q$, if $P \,\R\, Q$ for some bisimulation $\R$. 
\end{definition}

\medskip

We will use two subcalculi: the Internal $\pi$-calculus (\Intp),
and the Asynchronous Local $\pi$-calculus (\alpi), obtained by placing certain constraints
on prefixes. 

\paragraph{\Intp.}

In \Intp, all outputs are bound. This is 
syntactically enforced  by replacing the output construct with 
the bound-output construct $\bout a \tilb .P$, which, with respect to
 the grammar
of the ordinary  $\pi$-calculus, is an abbreviation for $\res \tilb
\out a \tilb . P$. In all tuples (input, output, abstractions, applications) the
components are pairwise distinct so to make sure that distinctions among names are
preserved by reduction. 



\txthere{Forwarders}{How to fix forwarders for internal pi: infinite forwarders. We need also:\\
We consider a sorting system over the channel names, so that a name carries its sort (and therefore the number of names that are transmitted over it, as well as whether the channel is linear or not).
}

\paragraph{\alpi.}

\alpi\ is defined by enforcing that in an input $\inp a\tilb.P$, all
names in $\tilb$ appear only in output position in $P$.
Moreover, 
 \alpi{} being  \emph{asynchronous},
  output prefixes have no continuation; 
in the grammar of the $\pi$-calculus this corresponds to having only outputs of the form 
 $\out a\tilb.\nil$ (which we will simply write $\out a\tilb$).
In   
\alpi, to maintain the  characterisation of  barbed congruence as (ground) bisimilarity,
the transition system has to be modified ~\cite{localpi}, 
 allowing the dynamic introduction of additional processes (the
  `links', sometimes also
called forwarders). 
   




\begin{theorem}
\label{t:bisbc}
\begin{enumerate}
\item
 In  \Intp, on agents that are image-finite up to~\bsim,  barbed congruence  and
bisimilarity 
coincide. 
\item\label{bisbc:alpi}
In  \alpi, on agents that are image-finite up to~\bsim  
and where no free name is used in input, 
barbed congruence  and
bisimilarity 
coincide. 
\end{enumerate}
 \end{theorem} 

All    encodings of the $\lambda$-calculus (into  \Intp\ and  \alpi) in the
paper satisfy the conditions of Theorem~\ref{t:bisbc}. 
Thus  we  will be able to use
bisimilarity as a proof technique for barbed congruence. 
(In part (2) of the theorem, the condition on inputs can be removed by
adopting an asynchronous variant of bisimilarity; however, the
synchronous version is easier to use in our proofs based on unique
solution of equations).


\section{Unique solutions in  \Intp\ and \alpi}
\label{s:usol}

We adapt   the proof technique of unique solution of equations, from~\cite{usol} to
the calculi  \Intp\ and \alpi, in order to  derive bisimilarity results. 
The technique is discussed in ~\cite{usol} on the asynchronous $\pi$-calculus 
(for possibly-infinite systems of equations).
The structure of the 
 proofs for \Intp\ and \alpi\ is similar; 
 in particular the completeness part 
is essentially the same 
 because
 bisimilarity  is the same. 
 The differences in the syntax of \Intp, and in the transition
 system of \alpi, show up only in certain technical details of the
 soundness 
 proofs. 

%



We need variables to write equations. We  use
 capital
letters  $X,Y,Z$
 for  these variables and call them \emph{equation variables}.
 The body of an equation is a name-closed abstraction
possibly containing equation  variables 
(that is,  applications can also be of the form $\app X\tila$). 
%
We use $E$ to range over such  expressions; and 
 $\EE$ to range over systems of equations, defined as follows. 
 In the definitions below, the indexing set $I$ can be infinite.

\begin{definition}
Assume that, for each $i$ of 
 a countable indexing set $I$, we have a variable $X_i$, and an expression
$E_i$, possibly containing  some variables. 
Then 
$\{  X_i = E_i\}_{i\in I}$
(sometimes written $\til X = \til E$)
is 
  a \emph{system of equations}. (There is one equation for each
  variable $X_i$; we sometimes use $X_i$ to refer to that equation.) 

  A system of equations is \emph{guarded} if each
  occurrence of a variable in the body of an equation is underneath a
  prefix.
\end{definition}

$E[\til F]$  is the abstraction resulting from $E$ by
replacing each variable $X_i$   with the abstraction $F_i$ (as usual
 assuming
$\til F$ and $\til X$ have the same sort). 


\begin{definition}\label{d:un_sol}
Suppose  $\{  X_i = E_i\}_{i\in I}$ is a system of equations. We say that:
\begin{itemize}
\item
 $\til F$ is a \emph{solution of the 
system of equations  for $\bsim$} 
if for each $i$ it holds
that $F_i \bsim E_i [\til F]$.
\item The  system has 
\emph{a unique solution for $\bsim$}  if whenever 
$\til F$ and $\til G$ are both solutions for $\bsim$, we have 
$\til F \bsim \til G$. \end{itemize}
 \end{definition} 

 \begin{definition}[Syntactic solutions]
   The syntactic solutions of the system of equations $\til X =\EeqBody{}{}$ are the  
   recursively defined constants $\KEi E \Defi E_i[\KE]$, for each
   $i\in I$, where $I$ is the indexing set of the system.
 \end{definition}
The syntactic solutions
of a system of equations 
are indeed solutions of
it. 
%




  A process $P$ \emph{diverges} if it can perform an infinite sequence
  of internal moves, possibly after some visible ones 
  (i.e., actions  different from $\tau$); formally,  there are
  processes $P_i$, $i\geq 0$, and some $n$, such that
  $P=P_0\arr{\mu_0} P_1 \arr{\mu_1} P_2 \arr{\mu_2}\dots$ and for all
  $i>n$, $\mu_i=\tau$. We call a \emph{divergence of $P$} the sequence
  of transitions $\big(P_i\arr{\mu_i}P_{i+1}\big)_{i}$.
In the case of an abstraction, 
$F$ has a divergence
if the process $\app F\tila$ has a divergence, where $\tila$ are fresh
names. A tuple of agents $\til A$ \emph{is divergence-free} if none of the
components $A_i$ has a  divergence. 

The following result is the technique we rely on to establish
completeness of the encoding. As announced above, it holds in both
\Intp\ and \alpi.
\begin{theorem}\label{thm:usol}
In \Intp\ and \alpi,
a guarded system of equations 
with divergence-free  syntactic 
solutions  has unique solution for \bsim. 
\end{theorem}

Techniques for ensuring termination, hence divergence freedom, for the
$\pi$-calculus have been 
studied in, e.g., \cite{termination1,termination2,termination3}.

\subsection{Further Developments}


We  present some further developments to the theory of unique solution of equations, 
that are needed  for the results in this paper. 
The first result allows us to derive the unique-solution property for a system of
equations from the analogous property of an extended system.

\begin{definition}
\label{d:extend}
A system of equations $\Eeq'$ \emph{extends} system $\Eeq$
if there exists a fixed set of indices $J$ such that any solution of
$\Eeq$ can be obtained from a solution of $\Eeq'$ by removing the
components corresponding to indices in $J$. 
\end{definition} 


\begin{theorem}
  \label{t:transf:equations}
  Consider two systems of equations 
$\Eeq'$ 
and $\Eeq$ 
where 
  $\Eeq'$
extends $\Eeq$. 
If $\Eeq'$ has a unique solution, then the property also holds for
$\Eeq$. 
\end{theorem}

We shall use Theorem~\ref{t:transf:equations} in
Section~\ref{s:complete}, in a situation where we transform a certain
system into another one, whose uniqueness of solutions
is easier to establish. 

\begin{remark}
We cannot derive Theorem~\ref{t:transf:equations} by comparing the syntactic solutions of
the two systems  $\Eeq'$ and $\Eeq$.
For instance,  the equations
  $X=\tau.X$ and $X=\tau.\tau.\tau\dots$ have (strongly) bisimilar   syntactic 
  solutions, yet only the latter equation has the unique-solution property.
(Further, Theorem~\ref{t:transf:equations} allows us to compare systems
of different size.) 
\end{remark}

The second development is a  generalisation of Theorem~\ref{thm:usol}
to preorders; we postpone its presentation to Section~\ref{s:contextual}.

\section{Milner's encodings}
\label{s:enc:cbv}
\subsection{Background}


Milner noticed \cite{milner:inria-00075405,encodingsmilner} that his
call-by-value 
 encoding can be easily tuned so to mimic
forms of 
evaluation in which, in an application $MN$, the function $M$ is run
first, or the argument $N$ is run first, or function and argument are
run in parallel (the proofs are actually carried out for this last
option). We chose here the first one, because it is more in line with
ordinary call-by-value. 
A discussion on the `parallel' call-by-value is deferred to
Section~\ref{s:concl}.

The core of any  encoding of the  $\lambda$-calculus  into a process calculus is the 
translation of function application. This  
becomes a particular form of 
parallel combination of two processes, the function and its argument;  
 $\betav$-reduction  is then  modeled as   process interaction.

The encoding of a $\lambda$-term is parametric over a name; this may
be thought of as the  \emph{location} of that term, or as its
\emph{continuation}.  
A term that    becomes a value   signals so at its continuation name
and, in doing so, 
it grants 
access 
 to the body of the
value. Such body   is replicated, so that the value may be
copied several times. When the value is a function, its body can
receive two names: (the access to) its value-argument, and the
following continuation.
In the translation of application, first the function is run, then
the argument; finally the function is informed of its argument and 
continuation.

In the  original paper~\cite{milner:inria-00075405},
Milner presented
two candidates for the encoding of call-by-value
$\lambda$-calculus~\cite{DBLP:journals/tcs/Plotkin75}.  They follow the same idea of
translation, but with a technical difference in the rule for
variables.  One encoding, $\qencm$, is so defined: 
\[ 
\begin{array}{rcl}
\encm{\abs xM} & \deff & 
\bind p\outb p y.!\inp y{x,q}.\encma Mq \\[\mypt]
\encm{MN}& \deff &  \\[\myptSmall]
\multicolumn{3}{r}{   
\bind p(\new q)(\encma Mq
|\inp qy.\new r(\encma Nr| \inp rw.\out
  y{w,p})) 
}\\[\mypt]
\encm{x}& \deff & \bind p \out p x
\end{array}
 \]
In the other encoding, $\qencmp$,
application and $\lambda$-abstraction are treated as in 
$\qencm$; the rule for variables is:
$$
\encmp{x}  \deff 
\bind p \outb p y.!\inp y {z,q}.\out x{z,q}
\enspace.
  $$



The encoding $\qencm$ is more efficient than $\qencmp$,
\iffull
\else
as it uses fewer communications.
\fi

\subsection{Some problems with the encoding}
\label{ss:ot}

The immediate free output in the encoding of variables in  $\qencm$ 
breaks the validity of $\betav$-reduction; i.e., there exist a term $M$
and a value $V$ such that $\encm{(\lambda x. M)V } \not\bsim
\encm{ M \sub V x }$~\cite{sangiorgiphd}. 
The encoding $\qencmp$ fixes 
this by communicating, instead 
of a
free name,  a 
fresh  pointer to that name. 
 Technically, the initial free output of $x$ is replaced by a
 bound output coupled with a link to $x$ (the process
 $!\inp y {z,q}.\out x{z,q}$, receiving at $y$ and re-emitting at $x$).
 Thus  $\betav$-reduction is validated~\cite{sangiorgiphd}.
%
(The final version of Milner's paper~\cite{encodingsmilner},
was
written after the results in~\cite{sangiorgiphd} were known and presents
only the encoding $\qencmp$.)

Nevertheless,      $\qencmp$ only delays the free output, as the added link
contains itself a free output. 
As a consequence, we can show that other desirable equalities of
call-by-value are broken. An example is law~\reff{eq:nonlaw} from the
Introduction, as stated by Proposition~\ref{p:nonlaw} below.
 This law is desirable (and indeed valid for contextual equivalence,
 or the Eager-Tree equality) 
intuitively because, in any substitution closure 
of the law,  either both terms diverge, or they
 converge to  the same value. The same argument holds 
 for their $\lambda$-closures, $\abs x x\val$ and $\abs x I(x\val)$. 
We recall that $\wbc\pi$ is barbed congruence in the 
$\pi$-calculus.
\begin{proposition}\label{p:nonlaw}
For any value $\val$, we have:
$$\encmp {I(x\val)} 
\nwbc\pi
\encmp {x\val} 
\mbox{ and }
\encm {I(x\val)}
\nwbc\pi
 \encm {x\val}
\enspace.
$$
\end{proposition}
(The law is
violated also under coarser equivalences, such as
contextual  equivalence.) 
Technically, the reason why the law fails in $\pi$ can be illustrated
when $\val=y$, for encoding $\qencm$. We have:
\begin{alignat*}{3}
  \encma {xy} p &\wbc\pi \outb x {v} .\new {w}&&(\out v {w, p}|!\inp w u.\out y
                  u)
                    \mbox{ \hspace{1ex}}
                    \\[\mypt]
  \encma {I(xy)} p&\wbc\pi \outb x v.\resb {w,q}&&(\out v {w, q} | !\inp{w} u.\out y
                       u
    \\
           &  &&  | 
                { \inp q z.\outb p {z'}.!\inp {z'} {w'}.\out z {w'}})     
\end{alignat*}
In presence
of the normal form $x y$, the identity $I$ becomes observable. Indeed, in the second
term, a fresh name, $q$, is sent instead of continuation $p$, and a
link between $q$ and $p$ is installed. This corresponds to a
law
which is valid in 
\alpi, but not in $\pi$.


This problem can be avoided by iterating the transformation that takes us
from $\qencm $ to $\qencmp$ (i.e., the replacement of a free output
with a bound output so to avoid all emissions of free names). Thus the
target language becomes Internal $\pi$; the resulting encoding is
analysed in Section~\ref{s:enc:pii}.

Another solution is to control the use of name capabilities in
processes. In this case the target language becomes \alpi,
 and we need not modify
 the  initial encoding  $\qencm$. This
situation is analysed in Section~\ref{s:localpi}.


Moreover, in both solutions,
the use of  link processes
validates the following law~---
a form of $\eta$-expansion~---
 (the law fails  for Milner's encoding into the $\pi$-calculus): 
 \[
\abs y x y = x
\]
In the call-by-value  $\lambda$-calculus this is a useful law 
(that holds because  substitutions  replace variables with values).

\section{Encoding in the Internal \pc}
\label{s:enc:pii}
\subsection{Encoding and soundness} 
\label{ss:enc_pii}

\begin{figure}[t]
  \begin{mathpar}
    \begin{array}{rcl}
\enca {\abs x M}&\deff& \bind p \outb p y . !\inp y {x,q}.\enc M q
\\[\mypt]
\enca x& \deff& \bind p \outb p y. \fwd y x
\\[\mypt]
      \enca {MN} &\deff& \bind p \new q \big(\enc M q
                         | \inp q y. \new r \big(\enc N r |
      \\[\myptSmall]
                \multicolumn{3}{r}{  ~~ \inp r w .\outb y {w',p'}.(\fwd {w'} w|\fwd {p'} p)\big)\big)}
    \end{array}
  \end{mathpar}
\caption{The encoding into \Intp} 
\label{f:enc_internal}
\end{figure}

Figure~\ref{f:enc_internal} presents the encoding into \Intp, derived
from Milner's encoding by removing the 
free outputs as explained in Section~\ref{s:enc:cbv}.
Process   $\fwd ab$  represents a \emph{link} (sometimes called forwarder; 
for readability we have adopted the infix
notation $\fwd ab$ for the constant $\fwd{}{}$). It 
transforms all outputs at $a$ into outputs at $b$ (therefore $a,b$ are
names of the same sort). Thus the body of $\fwd ab$ is replicated,
unless $a$ and $b$ are \emph{continuation names} (names such as
$p,q,r$ over which the encoding of a term is abstracted). 
The
definition of the constant $\fwd{}{}$ therefore is: 
$$
\begin{array}{rcl}
\fwd{}{} &\Defi &
\left\{ \begin{array}{l}
\bind{p,q} \inp p {x}.\outb {q} {y}
          . {\fwd{ y}{ x}}\\[\myptSmall]
\multicolumn{1}{r}{  
       ~\quad  \mbox{if $p,q$ are continuation names}
                   }              \\[\mypt]
\bind{x,y} ! \inp x {p,z}.\outb y {q,w}
          .({\fwd{ q}{ p}}|\fwd w z)  \\[\myptSmall]
\multicolumn{1}{r}{  
          ~\quad \mbox{otherwise}
}\end{array} \right. 
\end{array}
 $$
(The distinction between continuation names and the other sorts of
names is not necessary, but simplifies the proofs.)

The encoding   validates $\betav$-reduction.

\begin{lemma}[Validity of  $\betav$-reduction]\label{l:beta}For any
  $M,N$ in $ \Lao$, $M\longrightarrow  N$ implies 
$\enca M\bsim\enca N$.
\end{lemma}

The structure of the 
proof of soundness of the encoding is similar to that for the
analogous property for 
Milner's call-by-name encoding with respect to Levy-Longo Trees \cite{cbn}. 
The details are however different, as in call-by-value both the encoding 
and the trees (the Eager Trees extended to handle $\eta$-expansion) are
more complex.  

We first  need to establish an operational
correspondence for the encoding.
For this we make use of 
an optimised encoding,
obtained  from
the one 
 in Figure~\ref{f:enc_internal} by performing a few (deterministic) reductions, 
at the price of 
 a more complex definition.  Precisely, in the encoding of
 application, 
we remove some of the initial
communications, including those with which a term signals that it has 
become a value.  Correctness of the optimisations  is established
by algebraic reasoning. 


Using the operational correspondence, we then show that the observables 
for bisimilarity in the encoding $\pi$-terms  imply the observables for 
$\eta$-eager normal-form bisimilarity in the encoded $\lambda$-terms.
The delicate cases are those
in which a branch in the tree of the terms is produced~---
case \reff{ie:split} of Definition~\ref{enfbsim}~--- and where
 an $\eta$-expansion    occurs~---  thus
a variable is equivalent to an abstraction,  
 cases~\reff{lab:five} and~\reff{def:enfe:case:eta} of Definition~\ref{enfebsim}. 

For the branching,  we exploit a decomposition property on $\pi$-terms, roughly allowing
us to derive from the bisimilarity of two parallel compositions the componentwise 
bisimilarity of the single components. 
 For the $\eta$-expansion, 
if $\enca x \bsim\enca {\abs zM}$,
where   $M\converges \evctxt[x\val]$, 
 we use a coinductive argument to
derive $\val\enfe z$ and $\evctxt [y]\enfe y$, for $y$ fresh; 
from this we  then obtain
 $\abs zM \enfe x$.  




\begin{lemma}[Soundness]\label{l:sound}
For any  $M,N \in \Lao$, if $\enca M\bsim\enca N$ then $M\enfe N$.
\end{lemma}

\subsection{Completeness and Full Abstraction}
\label{s:complete}
To ease the reader into the proof, we first show the completeness for $\enf$, rather than $\enfe$.

\paragraph{The system of equations.}
Suppose $\R$ is an eager normal-form bisimulation. 
We
define a (possibly infinite) system of equations $\eqcbv$, solutions of which will be 
obtained from the encodings of the pairs in $\R$.
We then use Theorem~\ref{thm:usol} and Theorem~\ref{t:transf:equations}
to show that $\eqcbv$ has a unique solution.

We assume an ordering on names and variables, so to be able
to view (finite) sets of these as  tuples. 
Moreover, if $F$ is an abstraction, say $\bind \tila P$, then  $\bind \tily
F$  is an abbreviation for its uncurrying $\bind{\tily,\tila}P$.

There is one equation $ X_{M,N} = E_{M,N}$ for each pair $(M, N)\in\R$.
The body $E_{M,N}$ is essentially the encoding
of the eager normal form
 of $M$ and $N$,  with the
variables of the equations representing the coinductive hypothesis.
To formalise this, 
we extend the encoding of the $\lambda$-calculus to equation variables 
by setting
\[ 
 \enca {X_{M,N}}{} \deff \bind {p}  \app{X_{M,N}}{\tily,p} 
 \hskip .5cm \mbox{ ~~ where $\tily = \fv{M,N}$}
 \enspace.
\]
We now describe the equation $ X_{M,N} = E_{M,N}$, for
 $(M,N)\in\R$.
The equation is 
 parametrised on  the free variables
of  $M$ and $N $ (to ensure that the body $E_{M,N}$ 
 is a name-closed
abstraction) and an additional continuation
name (as all encodings of terms). Below $\tily =\fv{M,N}$.
\begin{enumerate}
\item If $M\converges x$ and $N\converges x$, then
 the equation
 is 
 the encoding of $x$:
\begin{align*}
X_{M,N}&=\bind {\til y} \enca x \\[\myptSmall]
&= \bind {\til y,p}  \outb p z. \fwd z x
\end{align*}


\item 
If 
$M\diverges$ and
  $N\diverges$, then the equation uses a purely-divergent term;
we choose  the encoding of  $\Omega$: 
\begin{align*}
X_{M,N} = \bind {\til y} \enca \Omega 
\end{align*}
%
\item If $M\converges \abs x M'$ and $N\converges \abs x N'$, then 
the equation encodes an abstraction whose body refers 
to the normal forms of  $M',N'$, via the variable $X_{M',N'}$:
\[
\begin{array}{rcl}
 X_{M,N}&=&\bind {\til y}
\enca{\abs x X_{M',N'}
}
\\[\myptSmall]
&= & \bind {\til y,p}
\outb p z .!\inp z
{x,q}.X_{M',N'}\param {\tilprime y,q}
\end{array}
\]





\item\label{item:decomp:eqcbv} If $M\converges \evctxt [x\val]$ and
  $N\converges \evctxt' [x\valp]$, 
  we separate the evaluation contexts and the values, as in 
  Definition~\ref{enfbsim}. 
  In the body of the equation, this is achieved by: $(i)$ rewriting
  $\evctxt[x\val]$ into $(\abs z\evctxt [z])(x\val)$, for some fresh $z$,
  and similarly for $\evctxt'$ and $\valp$ (such a transformation is
  valid for $\enf$); and $(ii)$ referring to the variable for the
  evaluation contexts, $X_{\evctxt[z],\evctxt'[z]}$,
and to the variable for the values,  $X_{\val,\valp}$.
This yields the equation (for $z$ fresh):
\begin{align*}
 X_{M,N} = \bind {\til y} 
\enca {(\abs z X_{\evctxt[z],\evctxt'[z]}
)
~(x~X_{\val,\valp}
  )}
\end{align*}
\end{enumerate}

As an example, 
 suppose 
 $(I,\abs
 xM)\in\R$, 
where $I=\abs x x$ and $M= (\abs {zy}z) x x'$. 
The
 free variables of $M$ are $x$ and $x'$. 
We  obtain
the following equations: 
\begin{enumerate}
\item $\begin{aligned}[t]
X_{I,\abs x M}
&=\bind {x'}\enca {\abs x X_{x,M}
} 
\\
&=\bind {x',p} \outb p y .!\inp y {x,q}.X_{x,M}\param {x,x',q}
\end{aligned}$
\item $\begin{aligned}[t]
X_{x,M
}
&=\bind {x,x'}\enca x \\ 
&=\bind{x,x',p} \outb p y.\fwd y x
\end{aligned}$
\end{enumerate}

\paragraph{Solutions of \eqcbv.}
Having set the system of  equations for $\R$, we now  define 
solutions for it from the encoding of the pairs in $\R$.

We can view the relation $\R $  as an ordered
sequence of pairs (e.g., assuming some lexicographical  ordering). 
Then $\R_i$ indicates the tuple obtained by
projecting the   pairs in $\R$ onto the $i$-th component ($i=1,2$).
Moreover $(M_j,N_j)$  is the $j$-th pair in $\R$, and $\til{y_j}$ is 
$\fvars {M_j,N_j}$.

We write $\encOO\R$   for the closed abstractions 
resulting  from the
 encoding of $\R_1$, i.e., the tuple 
whose $j$-th component is 
$ 
\bind {\til{y_j}}\enca{M_j}$, and similarly for 
$\encTO\R$.

\begin{lemma}\label{l:sol}
 $\encleft{\R}$
 and $\encright{\R}$
 are solutions of 
 \eqcbv.
\end{lemma}

\begin{proof}
We show that each component of  $\encOO\R$ is solution of the
corresponding equation, i.e., 
for the $j$-th component 
we show 
$\bind { \til{y_j}} \enca{M_j}\bsim \Eqsing{M_j,N_j} {\encOO {\R}}$.

We reason by cases over the shape of the 
 eager normal form of $M_j,N_j$. 
The most interesting case is when $M_j\converges \evctxt [x\val]$, in which
case we use the following equality (for $z$ fresh), which is proved
 using
 algebraic reasoning:
\begin{equation}  \label{eq:l:solMAIN}
\enca {(\abs z\evctxt [z])(x\val)} \bsim  \enca {\evctxt[x\val]}
\enspace.
\end{equation}
%
We also exploit the validity of $\betav$ for $\bsim$ (Lemma~\ref{l:beta}).
\end{proof}

\paragraph{Unique solution for \eqcbv.}
We use Theorem~\ref{t:transf:equations} to prove uniqueness 
of solutions for \eqcbv. 
The only delicate requirement is the one on divergence for the syntactic
solution.  
We introduce for this
 an auxiliary system of equations,  \eqcbvp, that extends \eqcbv, and
 whose syntactic solutions have no $\tau$-transition and hence trivially
 satisfy the requirement.  
Like the original system 
\eqcbv, so the new one 
 \eqcbvp\ is defined by inspection of the pairs in  $\R$;  
in 
 \eqcbvp, however, a pair of $\R$ may sometimes yield more than one
 equation. 
Thus, let $(M,N)\in \R$ with $\til y=\fvars {M,N}$.
\begin{enumerate}
\item When $M\Uparrow$ and  $N\Uparrow$, the equation is
\begin{align*}
  X_{M,N} = \bind{\til y,p} \zero
  \enspace.
\end{align*}
\item When $M\converges \val$ and $N\converges \valp$, we introduce a
  new  equation variable $\XV_{\val,\valp}$  and a new equation; 
this will allow us, in the following step (3), to 
perform some  optimisations. The equation is
%
\begin{align*}
 X_{M,N} = \bind {\til y,p} \outb p z .  \XV_{\val,\valp}\param{z,\tilprime y}
  \enspace,
\end{align*}
and we have, accordingly, the two following additional equations
corresponding to the cases where values are functions or variables:
%
%
\[
\begin{array}{rcl}
\XV_{\abs x M',\abs x N'}&=& \bind {z,\til y}!\inp z {x,q}
  . X_{M',N'}\param{{\tilprime y},q}
\\[\mypt]
\XV_{x,x}&=&\bind {z,x} \fwd z x
\end{array} 
\]

\item When $M\converges \evctxt[x\val]$ and $N\converges\evctxt[x\valp]$, we
refer 
to $\XV_{\val,\valp}$, instead of $X_{\val,\valp}$, so  to remove all
initial reductions in the corresponding equation for 
\eqcbv.  The first action thus becomes 
  an output: 
\[
  \begin{array}{rcl}
  X_{M,N} &=&  \\[\myptSmall]
\multicolumn{3}{r}{ 
\bind {\til y,p} \outb x
              {z,q}.(\XV_{\val,\valp}\param{z,{\tilprime y}}
 |\inp q
       w.X_{\evctxt[w],\evctxt'[w]}\param{{\tilpprime y},p})
}
  \end{array}
\]
\end{enumerate}



Lemmas~\ref{l:div_aux} and~\ref{l:div} 
are needed to apply Theorem~\ref{t:transf:equations}.
(In the statement of Lemma~\ref{l:div_aux},  `extend' is as by Definition~\ref{d:extend}.)

\begin{lemma}\label{l:div_aux}The system of equations \eqcbvp extends 
the system of equations \eqcbv.
\end{lemma}
\begin{proof}
The new system 
\eqcbvp\ is obtained from  \eqcbv\  
by modifying the equations  and 
adding new ones.
Ones shows that the solutions to the common equations are the same, 
using  algebraic reasoning.
\end{proof}

\begin{lemma}\label{l:div}
 \eqcbvp has a unique solution.
\end{lemma}
\begin{proof}
Divergence-freedom for the
 syntactic solutions
  of \eqcbvp\ 
holds because in
the equations each name (bound or free) can appear either only in inputs
or only in outputs. As a consequence, since the labelled transition system is ground (names
are only replaced by fresh ones), no $\tau$-transition can ever be
performed, after any number of visible actions. 
Further,  \eqcbvp is  guarded. Hence   we can apply Theorem~\ref{thm:usol}.
%
%
\end{proof}

\begin{lemma}[Completeness for $\enf$]\label{l:complete_enf}
  $M\enf N$ implies $\enca M \bsim \enca N$,
for any $M,N \in\Lao$. 
\end{lemma}
 
\begin{proof}
Consider 
  an eager normal-form bisimulation $\R$, and 
the corresponding  systems of equations \eqcbv\ and  \eqcbvp.
  Lemmas~\ref{l:div} and~\ref{l:div_aux} allow us to apply
  Theorem~\ref{t:transf:equations} and deduce that \eqcbv\ has a unique
  solution.
By Lemma \ref{l:sol},
  $\encOO \R$ and $\encTO\R$ are solutions of \eqcbv. 
Thus, from $M\RR N$,
 we deduce
$\bind {\til y} \enca M  \bsim \bind{\til y} \enca N$, 
  where $\til y=\fvars {M,N}$. 
Hence also 
$ \enca M  \bsim \enca N$.



\end{proof}

\paragraph{Completeness for $\enfe$.} 
The proof  for $\enf$ %
is extended to $\enfe$, maintaining  its
 structure. We highlight the main differences.




We enrich \eqcbv\ with the equations corresponding to the two
additional clauses of $\enfe$ (Definition~\ref{enfebsim}).
When $M\converges x$ and $N\converges \abs z N'$, where $N'\enfe xz$, 
we proceed as in case~\ref{item:decomp:eqcbv} of the definition of \eqcbv,
 given that  
$N\enfe \abs z \left( (\abs w\evctxt [w]) (x\val)\right)$;  the
equation is:
\begin{align*}
X_{M,N} =\bind {\til y}
  \enca{\abs z \left((\abs w X_{w,\evctxt[w]}) ~ (x~X_{z,\val}
  )\right)}
  \enspace.
\end{align*}
%
We proceed likewise for the symmetric case.


In the optimised equations that we use to derive unique solutions, 
we add the following equation (relating
values),  as well as its symmetric counterpart:
%
\begin{align*}
        \begin{array}{rcl}
\XV_{x,\abs z N'}
&=&\bind {y_0,\til y}
\\
\multicolumn{3}{r}{ 
    !\inp {y_0}{z,q}. \outb x {z',q'}. (\XV_{z,\val}\param{z',\tilprime y }
         |\inp {q'}{w}.X_{w,\evctxt[w]}\param{\tilpprime y,q})
         \enspace.
}        \end{array}
\end{align*}

Finally, 
to
prove that $\encOO \R$ and $\encTO \R$ are 
 solutions of \eqcbv, we 
 show that, whenever $M\converges x$ and $N\converges \abs z N'$, 
 with $N'\converges\evctxt[x\val]$:
\begin{align*}
 \enca M  &\bsim \Eqsing{M,N} {\encleft\R}\param{\til y}\\[\myptSmall]
          &= \enca {\abs z\left( (\abs w w)(x z)\right)}
            \quad\mbox{}
\end{align*}
\mbox{and} 
\begin{align*}
 \enca N  &\bsim \Eqsing{M,N} {\encright\R}\param{\til y}\\[\myptSmall]
          &= \enca {\abs z\left( (\abs w \evctxt [w])(x \val)\right)}
            \enspace.
\end{align*}
To establish the former, we 
use algebraic reasoning to infer
 $\enca x \bsim \enca {\abs z xz}$. 
For the latter, we use law~\reff{eq:l:solMAIN} (given in the proof
of Lemma~\ref{l:sol}).

\begin{lemma}[Completeness for $\enfe$] \label{l:complete}
  For any $M,N$ in $\Lao$, 
  $M\enfe N$ implies $\enca M \bsim \enca N$.
\end{lemma}

Combining Lemmas~\ref{l:sound} and~\ref{l:complete}, and
Theorem~\ref{t:bisbc} we derive Full Abstraction for  
$\enfe$ with respect to barbed congruence. 

\begin{theorem}[Full Abstraction for $\enfe$]
For any $M,N$ in $\Lao$, we have $M\enfe N$ iff $\enca M\wbc\IntpSmall \enca N$
\end{theorem}

\begin{remark}[Unique solutions versus up-to techniques]
\label{r:upto}
For Milner's encoding of call-by-name $\lambda$-calculus, 
the completeness part of the full abstraction result with respect to
L{\'e}vy-Longo Trees~\cite{cbn} relies on
\emph{up-to techniques for bisimilarity}. 
Precisely, given a relation $\R$ on $\lambda$-terms that represents a
tree bisimulation,  one shows that the $\pi$-calculus encoding of 
$\R$  is a $\pi$-calculus bisimulation \emph{up-to context and
  expansion}. Expansion is a preorder that intuitively guarantees that a
term is `more efficient' than another one. 
In the up-to technique, expansion is used
to manipulate the derivatives of two transitions so to bring up a
common context.   
Such up-to technique is not powerful enough for the 
 call-by-value encoding and the Eager
Trees
 because 
some of the required transformations would violate expansion (i.e., they
would require to replace a term by  a `less efficient' one).
An example of this is law~\reff{eq:l:solMAIN} (in the proof
of Lemma~\ref{l:sol}), that would have to be
applied from right to left so to 
implement the branching in clause
\reff{ie:split} of Definition~\ref{enfbsim}
(as a context with two holes). 

The use of the technique of  unique solution of equations allows us to
overcome the problem:  law \reff{eq:l:solMAIN} and similar laws that
introduce 'inefficiencies' can be used (and they are indeed used, 
 in various places),  
as long as they 
 do not produce  new divergences.
\end{remark}

\section{Encoding into  \alpi}
\label{s:localpi}

Full abstraction
with respect to  $\eta$-Eager-Tree equality also holds
 for Milner's simplest encoding, namely $\qencm$
 (Section~\ref{s:enc:cbv}),
provided that 
  the target
language of the encoding is taken to be  \alpi. 
The adoption of  \alpi\ implicitly allows us to control capabilities,
avoiding violations of laws such as~\reff{eq:nonlaw} in 
the Introduction. 
In \alpi, bound output prefixes such as  $\outb a x .\inp x y$ are abbreviations for 
$\new x(\out a x| \inp x y )$.







\begin{theorem}
\label{t:faALpi}
$M\enfe N$ iff $  \encm M  \wbc\alpiSmall \encm N$, for any $M,N \in \Lao$. 
\end{theorem}

  The main difference 
with respect to  the proofs of Lemmas~\ref{l:complete_enf} 
and~\ref{l:complete}
  is when proving absence of divergences for the (optimised) system of equations.
%
 Indeed,  in \alpi\  the characterisation of barbed congruence 
($ \wbc\alpiSmall$) as bisimilarity 
makes use 
of 
 a different labelled transition system 
where visible transitions may create  
new processes (the `static links'), that could thus  produce new 
reductions. Thus one has to show that the added processes 
do not introduce 
new divergences. 

\section{Contextual equivalence and preorders}
\label{s:contextual}
We have presented full abstraction for $\eta$-Eager-Tree equality
taking
a `branching' behavioural equivalence, namely
 barbed congruence,  on the $\pi$-processes.
We show here the same result
for contextual
equivalence,  the most common 
 `linear' behavioural equivalence.
We also extend the results to   preorders.

We only discuss the encoding 
$\qenca$ into 
\Intp. Similar results however hold for 
the encoding $\qencm$ into
  \alpi.

\subsection{Contextual relations and traces}
\label{ss:preorders}
Contextual equivalence   is defined in the $\pi$-calculus 
analogously to its definition in the  
 $\lambda$-calculus (Definition~\ref{d:ctxeq}); 
thus, with
respect to barbed congruence, the bisimulation game on reduction is
dropped. 
Since we wish to handle  preorders, we also  introduce
the \emph{contextual preorder}. 

\begin{definition}
\label{d:ctx_pi}
Two \Intp\  agents 
$A,B$ are in the \emph{contextual preorder}, written 
$A \ctxpre B $, 
 if 
$C[A]\Dwa_{ a}$ implies $C[B]\Dwa_{ a}$, 
for all contexts $C$. They are 
\emph{contextually equivalent}, 
written 
$A \ctxeqPI B $, if both 
$A \ctxpre B $  and $B \ctxpre A $ hold.
\end{definition} 

To manage  contextual  preorder and equivalence in proofs, we
exploit  
 characterisations of them  as trace inclusion and  equivalence.
For $s = \mu_1, \ldots, \mu_n$, 
where each $\mu_i$ is a visible action, 
 we set $P \Arr{s} $
if $P \Arr{\mu_1} P_1 \Arr{\mu_2}P_2 \ldots P_{n-1} \Arr{\mu_n}P_n$, for
some  processes $P_1, \ldots, P_n$.

\begin{definition}
\label{d:trace}
Two \Intp\ processes  $P,Q$ are in the \emph{trace inclusion}, written 
$P \trincl Q $, if $P \Arr{s} $ implies $Q \Arr{s} $, for each trace
$s$. They are 
\emph{trace equivalent}, 
written 
$P \treq Q $, if both 
$P \trincl Q $ and 
$Q \trincl P $ hold.
\end{definition} 
As usual, these relations are extended to abstractions by requiring
instantiation of the parameters with fresh names. 

\begin{theorem}
\label{t:cha_tr}
In \Intp, relation  $ \ctxpre$ coincides with $ \trincl$, and 
relation $ \ctxeqPI$ coincides with $ \treq  $.
\end{theorem}





\subsection{A proof technique for preorders}
\label{ss:usol_pre}

We modify the technique of unique solution of equations to
reason  about preorders, precisely the 
 trace inclusion preorder. 

In the case of equivalence, 
the 
 technique of unique solutions   exploits symmetry arguments,
but 
symmetry  does not hold for preorders.  
We overcome the problem by referring to the
syntactic solution of the  system in an asymmetric manner. 
 This   yields the two lemmas below, intuitively stating  that the
 syntactic solution    of a  system
is its smallest pre-fixed point, as well as, under the  divergence-freeness
hypothesis, its greatest
post-fixed point. 
We say that $\til F$ is  a \emph{pre-fixed
   point 
 for $\trincl$} of  
a system of equations   $\{\til X= \til E\}$ 
if $\til E[\til F]\trincl \til F$;
similarly,  $\til F$ is  a \emph{post-fixed
   point  for $\trincl$}
if 
$\til F\trincl \til E[\til F]$.

\begin{lemma}[Pre-fixed points,  $\trincl$]\label{least}
Let $\Eeq$ be a  
system of equations, 
and $\KEE$ its syntactic solution. 
If  $\til F$ is 
a pre-fixed point  for $\trincl$ of $\Eeq$, 
then $\KEE \trincl \til F$
\end{lemma}

\begin{lemma}[Post-fixed points,  $\trincl$]\label{greatest}
Let $\Eeq$ be a guarded 
system of equations, 
and $\KEE$ its syntactic solution. 
Suppose $\KEE$ has no divergences. 
If  $\til F$ is 
a post-fixed point  for $\trincl$ of $\Eeq$, 
then
$\til F \trincl \KEE$.
\end{lemma}

Lemma \ref{least} is immediate; 
the proof of Lemma \ref{greatest}
is  similar to the 
proof of Theorem~\ref{thm:usol} (for bisimilarity). 
We thus derive the following proof technique. 

\begin{theorem}
\label{thm:usol_pre}
Suppose that $\Eeq$ is a guarded 
system of equations 
 with a 
divergence-free 
 syntactic
solution.
If $\til F$ is a pre-fixed point   for $\trincl$ of $\Eeq$, and 
 $\til G$  a post-fixed point, 
then
 ${\til F}\trincl {\til G}$. 
\end{theorem}


We can also extend
Theorem~\ref{t:transf:equations} to preorders.
We say that a system of equations $\Eeq'$ 
\emph{extends $\Eeq$ with respect to a given preorder 
}
 if there exists a fixed set of indices $J$ such that:
\begin{enumerate}
\item
  any pre-fixed point 
 of
$\Eeq$ for the preorder can be obtained from a pre-fixed point 
of $\Eeq'$ (for the same preorder) by removing the
components corresponding to indices in $J$;

\item 
the same as (1) with post-fixed points in place of pre-fixed points.
\end{enumerate}
 
\begin{theorem}\label{t:transf:preorders}
Consider two systems of equations $\Eeq'$ and $\Eeq$ 
where $\Eeq'$
extends $\Eeq$ with respect to $\trincl$. 
Furthermore, suppose $\Eeq'$ is  guarded and has a divergence-free 
syntactic solution. 
If $\til F$ is a pre-fixed point  for $\trincl$ of $\Eeq$, and 
 $\til G$  a post-fixed point, 
then
 ${\til F}\trincl {\til G}$. 
\end{theorem}

\subsection{Full abstraction results}
\label{ss:fa_preorders}

The preorder on $\lambda$-terms induced by the contextual preorder
 is  
 \emph{$\eta$-eager normal-form similarity}, $\esim$. It is   obtained by
imposing that $M \esim N$ for all $N$, whenever $M$ is  divergent.
Thus, with respect to the bisimilarity relation $\enfe$, we only have to change 
clause (1) of  Definition~\ref{enfbsim}, 
by  requiring 
 only $M$ 
to be divergent. (The bisimilarity  $\enfe$ is then the 
intersection of $\esim$ and its converse $\revesim$.)


\begin{theorem}[Full abstraction on preorders]
\label{t:preorders}
For any $M,N \in \Lao$, we have
$M\esim N$ iff $ \enca M \ctxpre \enca N$.
 \end{theorem} 

The structure of the proofs is similar to
 that for bisimilarity, 
 using however Theorem \ref{thm:usol_pre}. 
We discuss the main aspects of  the completeness part.

Given an \enfse $\R$, we define a system of equations \eqcbv as in
Section \ref{s:complete}. The only notable difference in the definition 
of the equations is in the case
where $M\R N$, $M$ diverges and $N$ has an \enform. In this case, we
use the following equation instead:
\begin{equation}\label{e:tr}
X_{M,N}=\bind {\til y} \enca \Omega
\enspace.
\end{equation}
As in Section \ref{s:complete}, we define a system of
guarded equations \eqcbvp\ 
 whose syntactic solutions do not diverge. 
Equation~\reff{e:tr} is replaced with $X_{M,N}=\bind {\til y,p} \zero$.

Exploiting Theorem~\ref{t:transf:preorders}, 
we can use unique solution for preorders
 (Theorem~\ref{thm:usol_pre}) 
 {with \eqcbv instead of \eqcbvp.} 

Defining $\encleft \R$ and $\encright \R$ as previously, we need to
prove that $\encleft\R \trincl\Eq \R{\encleft\R}$ and 
$\Eq \R {\encright\R}\trincl \encright\R$. The former result is
established along the lines of the analogous result in
Section~\ref{s:complete}: indeed, $\encleft\R$ is a solution of
\eqcbv for $\bsim$, and $\treq$ is coarser than $\bsim$. 

For the latter, the only difference is due to equation \reff{e:tr},
when $M\R N$, and $M$ diverges but not $N$.
In that case, we have to prove that $\enca \Omega\trincl \enca N$, which
follows easily because 
the only trace of $\enca\Omega$ is the empty one, hence
$\enc \Omega p \trincl P$ for any $P$.

\begin{corollary}[Full abstraction for $\ctxeqPI$] 
\label{t:faCTXpiI}
  For any $M,N$ in $\Lao$, 
  $M\enfe N$ iff $\enca M \ctxeqPI \enca N$.
\end{corollary}

\section{Conclusions and future work}
\label{s:concl}

In the paper we have studied the main question raised in Milner's
landmark paper on functions as $\pi$-calculus processes, which is
about the equivalence induced on $\lambda$-terms by their process 
encoding.  We have focused on call-by-value, where the problem was
still open; as behavioural equivalence on $\pi$-calculus we have 
taken contextual equivalence and  barbed congruence (the most common
`linear' and 'branching' equivalences). 

First we have shown that some expected equalities 
for open  terms fail under Milner's encoding. We have considered two
ways for overcoming  this issue: rectifying the encodings (precisely,
avoiding free outputs); restricting the target language to \alpi, 
so to
control the capabilities of exported names.
We have proved that, in both cases, the equivalence induced  is 
Eager-Tree equality, modulo $\eta$
(i.e., Lassen's \enfbsim). 
We have then introduced a preorder on these trees, and 
 derived similar full abstraction results  for them with respect to 
the contextual preorder on $\pi$-terms. 
The paper is also a test case for 
the technique of
unique solution of equations (and inequations),
which is essential in all our
 completeness proofs.

Lassen had introduced  Eager Trees as the call-by-value analogous of
L{\'e}vy-Longo and B{\"o}hm Trees. 
The results in the paper confirm the claim,  on process encodings of
$\lambda$-terms: it was known that for (weak and strong) call-by-name, the  equalities
induced are  those of  L{\'e}vy-Longo Trees and   B{\"o}hm Trees~\cite{xian}.

For controlling capabilities, we have used \alpi. 
Another possibility would have been to use a  type system. 
In this case however, the technique of unique solution of equations
needs to be extended to typed calculi. We leave this for future work.

We also leave for future work a thorough comparison between the technique
of unique solution of equations and  techniques  based on enhancements
of the bisimulation proof method (the ``up-to'' proof techniques),
including if and how our completeness results can be
derived using the latter techniques. (We recall that  the ``up-to''
proof techniques are used in the completeness proofs  with
respect to L{\'e}vy-Longo Trees and   B{\"o}hm Trees for the
\emph{call-by-name} encodings. 
We have discussed the problems with
call-by-value in Remark~\ref{r:upto}.) In any case, even if 
other solutions existed, for this specific problem the unique solution 
technique appears to provide an elegant and natural framework to 
carry out the proofs.

For our encodings  we have used the polyadic $\pi$-calculus; Milner's
original paper \cite{milner:inria-00075405} used the monadic calculus (the polyadic $\pi$-calculus
makes the encoding easier to read; it had not been introduced at the
time of \cite{milner:inria-00075405}). We believe that polyadicity
does not affect the results in the paper (the possibility of
autoconcurrency breaks full abstraction of the 
 encoding of the polyadic
$\pi$-calculus into the monadic one, but autoconcurrency does not
appear in the encoding of $\lambda$-terms).

In the call-by-value strategy we have followed, the function is
reduced before the argument in an application. Our results can be
adapted to the case in which the argument runs first, changing the
definition of evaluation contexts. The parallel call-by-value, in
which function and argument can run in parallel (considered in
\cite{encodingsmilner}), appears more delicate, as we cannot rely on the usual
notion of evaluation context.


 Interpretations of $\lambda$-calculi into 
 $\pi$-calculi 
appear  related to
game
semantics~\cite{BHYseqpi,DBLP:conf/fpca/HylandO95,DBLP:journals/tcs/HondaY99}.
In particular, for untyped call-by-name they both allow us to derive 
 B{\"o}hm Trees and  L{\'e}vy-Longo 
Trees~\cite{DBLP:journals/tcs/KerNO03,DBLP:journals/tcs/OngG04}.
 To our knowledge,
game semantics exist based on 
typed call-by-value,
e.g.,~\cite{DBLP:conf/csl/AbramskyM97,DBLP:journals/tcs/HondaY99}, but
not in the untyped case. 
In this respect, it would be interesting to see whether the
relationship between $\pi$-calculus and Eager Trees studied  in 
this paper could help to establish similar relationships in game
semantics. 



\begin{acks}
This work has been supported by the \grantsponsor{}{
European
Research Council (ERC)}{} under the Horizon 2020 programme (CoVeCe,
grant agreement No \grantnum{}{678157}); 
the  \grantsponsor{}{ANR}{} under the programmes 
``Investissements d'Avenir'' (\grantnum{}{ANR-11-IDEX-0007}),  \grantsponsor{}{LABEX MILYON}{}
(\grantnum{}{ANR-10-LABX-0070}), 
and \grantsponsor{}{Elica}{}  
 (\grantnum{}{ANR-14-CE25-0005}); and the
\grantsponsor{}{Universit{\'e} Franco-Italienne}{} under
the programme Vinci.
\end{acks}

\bibliography{efap}


\begin{thebibliography}{32}


\ifx \showCODEN    \undefined \def \showCODEN     #1{\unskip}     \fi
\ifx \showDOI      \undefined \def \showDOI       #1{#1}\fi
\ifx \showISBNx    \undefined \def \showISBNx     #1{\unskip}     \fi
\ifx \showISBNxiii \undefined \def \showISBNxiii  #1{\unskip}     \fi
\ifx \showISSN     \undefined \def \showISSN      #1{\unskip}     \fi
\ifx \showLCCN     \undefined \def \showLCCN      #1{\unskip}     \fi
\ifx \shownote     \undefined \def \shownote      #1{#1}          \fi
\ifx \showarticletitle \undefined \def \showarticletitle #1{#1}   \fi
\ifx \showURL      \undefined \def \showURL       {\relax}        \fi
\providecommand\bibfield[2]{#2}
\providecommand\bibinfo[2]{#2}
\providecommand\natexlab[1]{#1}
\providecommand\showeprint[2][]{arXiv:#2}

\bibitem[\protect\citeauthoryear{Abramsky}{Abramsky}{1987}]%
        {Abr88}
\bibfield{author}{\bibinfo{person}{Samson Abramsky}.}
  \bibinfo{year}{1987}\natexlab{}.
\newblock \showarticletitle{The Lazy $\lambda$-calculus}.
\newblock In \bibinfo{booktitle}{\emph{Research Topics in Functional
  Programming}}, \bibfield{editor}{\bibinfo{person}{D.~Turner}} (Ed.).
  \bibinfo{publisher}{Addison Wesley}, \bibinfo{pages}{65--117}.
\newblock


\bibitem[\protect\citeauthoryear{Abramsky and McCusker}{Abramsky and
  McCusker}{1997}]%
        {DBLP:conf/csl/AbramskyM97}
\bibfield{author}{\bibinfo{person}{Samson Abramsky} {and} \bibinfo{person}{Guy
  McCusker}.} \bibinfo{year}{1997}\natexlab{}.
\newblock \showarticletitle{Call-by-Value Games}. In
  \bibinfo{booktitle}{\emph{Proceedings of, {CSL} '97, Annual Conference of the
  EACSL, Selected Papers}}, Vol.~\bibinfo{volume}{1414}.
  \bibinfo{publisher}{Springer}, \bibinfo{pages}{1--17}.
\newblock


\bibitem[\protect\citeauthoryear{Accattoli and Guerrieri}{Accattoli and
  Guerrieri}{2016}]%
        {accattolicbv}
\bibfield{author}{\bibinfo{person}{Beniamino Accattoli} {and}
  \bibinfo{person}{Giulio Guerrieri}.} \bibinfo{year}{2016}\natexlab{}.
\newblock \showarticletitle{Open Call-by-Value}. In
  \bibinfo{booktitle}{\emph{Proc. of {APLAS} 2016}}
  \emph{(\bibinfo{series}{Lecture Notes in Computer Science})},
  Vol.~\bibinfo{volume}{10017}. \bibinfo{publisher}{Springer Verlag},
  \bibinfo{pages}{206--226}.
\newblock


\bibitem[\protect\citeauthoryear{Barendregt}{Barendregt}{1984}]%
        {barendregt1984lambda}
\bibfield{author}{\bibinfo{person}{H.P. Barendregt}.}
  \bibinfo{year}{1984}\natexlab{}.
\newblock \bibinfo{booktitle}{\emph{The lambda calculus: its syntax and
  semantics}}.
\newblock \bibinfo{publisher}{North-Holland}.
\newblock
\showISBNx{9780444867483}
\showLCCN{84005966}


\bibitem[\protect\citeauthoryear{Berger, Honda, and Yoshida}{Berger
  et~al\mbox{.}}{2001}]%
        {BHYseqpi}
\bibfield{author}{\bibinfo{person}{Martin Berger}, \bibinfo{person}{Kohei
  Honda}, {and} \bibinfo{person}{Nobuko Yoshida}.}
  \bibinfo{year}{2001}\natexlab{}.
\newblock \showarticletitle{Sequentiality and the pi-Calculus}. In
  \bibinfo{booktitle}{\emph{Proceedings of {TLCA}}}
  \emph{(\bibinfo{series}{Lecture Notes in Computer Science})},
  Vol.~\bibinfo{volume}{2044}. \bibinfo{publisher}{Springer},
  \bibinfo{pages}{29--45}.
\newblock


\bibitem[\protect\citeauthoryear{Demangeon, Hirschkoff, and
  Sangiorgi}{Demangeon et~al\mbox{.}}{2010}]%
        {termination2}
\bibfield{author}{\bibinfo{person}{Romain Demangeon}, \bibinfo{person}{Daniel
  Hirschkoff}, {and} \bibinfo{person}{Davide Sangiorgi}.}
  \bibinfo{year}{2010}\natexlab{}.
\newblock \showarticletitle{Termination in Impure Concurrent Languages}. In
  \bibinfo{booktitle}{\emph{Proc.\ 21th Conf.\ on Concurrency Theory}}
  \emph{(\bibinfo{series}{Lecture Notes in Computer Science})},
  Vol.~\bibinfo{volume}{6269}. \bibinfo{publisher}{Springer},
  \bibinfo{pages}{328--342}.
\newblock
\showISBNx{978-3-642-15374-7}


\bibitem[\protect\citeauthoryear{Durier, Hirschkoff, and Sangiorgi}{Durier
  et~al\mbox{.}}{2017}]%
        {usol}
\bibfield{author}{\bibinfo{person}{Adrien Durier}, \bibinfo{person}{Daniel
  Hirschkoff}, {and} \bibinfo{person}{Davide Sangiorgi}.}
  \bibinfo{year}{2017}\natexlab{}.
\newblock \showarticletitle{{Divergence and Unique Solution of Equations}}. In
  \bibinfo{booktitle}{\emph{Proceedings of {CONCUR} 2017}}
  \emph{(\bibinfo{series}{{LIPI}cs})}, Vol.~\bibinfo{volume}{85}.
  \bibinfo{publisher}{Schloss Dagstuhl--Leibniz-Zentrum fuer Informatik},
  \bibinfo{pages}{11:1--11:16}.
\newblock
\showISBNx{978-3-95977-048-4}
\showISSN{1868-8969}


\bibitem[\protect\citeauthoryear{Hindley and Seldin}{Hindley and
  Seldin}{1986}]%
        {DBLP:books/cu/HindleyS86}
\bibfield{author}{\bibinfo{person}{J.~Roger Hindley} {and}
  \bibinfo{person}{Jonathan~P. Seldin}.} \bibinfo{year}{1986}\natexlab{}.
\newblock \bibinfo{booktitle}{\emph{Introduction to Combinators and
  Lambda-Calculus}}.
\newblock \bibinfo{publisher}{Cambridge University Press}.
\newblock


\bibitem[\protect\citeauthoryear{Honda and Yoshida}{Honda and Yoshida}{1999}]%
        {DBLP:journals/tcs/HondaY99}
\bibfield{author}{\bibinfo{person}{Kohei Honda} {and} \bibinfo{person}{Nobuko
  Yoshida}.} \bibinfo{year}{1999}\natexlab{}.
\newblock \showarticletitle{Game-Theoretic Analysis of Call-by-Value
  Computation}.
\newblock \bibinfo{journal}{\emph{Theor. Comput. Sci.}} \bibinfo{volume}{221},
  \bibinfo{number}{1-2} (\bibinfo{year}{1999}), \bibinfo{pages}{393--456}.
\newblock


\bibitem[\protect\citeauthoryear{Hyland and Ong}{Hyland and Ong}{1995}]%
        {DBLP:conf/fpca/HylandO95}
\bibfield{author}{\bibinfo{person}{J.~M.~E. Hyland} {and}
  \bibinfo{person}{C.{-}H.~Luke Ong}.} \bibinfo{year}{1995}\natexlab{}.
\newblock \showarticletitle{Pi-Calculus, Dialogue Games and {PCF}}. In
  \bibinfo{booktitle}{\emph{Proceedings of {FPCA} 1995}}.
  \bibinfo{publisher}{{ACM}}, \bibinfo{pages}{96--107}.
\newblock


\bibitem[\protect\citeauthoryear{Ker, Nickau, and Ong}{Ker
  et~al\mbox{.}}{2003}]%
        {DBLP:journals/tcs/KerNO03}
\bibfield{author}{\bibinfo{person}{Andrew~D. Ker}, \bibinfo{person}{Hanno
  Nickau}, {and} \bibinfo{person}{C.{-}H.~Luke Ong}.}
  \bibinfo{year}{2003}\natexlab{}.
\newblock \showarticletitle{Adapting innocent game models for the B{\"{o}}hm
  treelambda-theory}.
\newblock \bibinfo{journal}{\emph{Theor. Comput. Sci.}} \bibinfo{volume}{308},
  \bibinfo{number}{1-3} (\bibinfo{year}{2003}), \bibinfo{pages}{333--366}.
\newblock


\bibitem[\protect\citeauthoryear{Lassen}{Lassen}{2005}]%
        {lassentrees}
\bibfield{author}{\bibinfo{person}{S{\o}ren~B. Lassen}.}
  \bibinfo{year}{2005}\natexlab{}.
\newblock \showarticletitle{Eager Normal Form Bisimulation}. In
  \bibinfo{booktitle}{\emph{20th {IEEE} Symposium on Logic in Computer Science
  {(LICS} 2005), 26-29 June 2005, Chicago, IL, USA, Proceedings}}.
  \bibinfo{publisher}{{IEEE} Computer Society}, \bibinfo{pages}{345--354}.
\newblock


\bibitem[\protect\citeauthoryear{Lassen and Levy}{Lassen and Levy}{2007}]%
        {lassentrees2}
\bibfield{author}{\bibinfo{person}{S{\o}ren~B. Lassen} {and}
  \bibinfo{person}{Paul~Blain Levy}.} \bibinfo{year}{2007}\natexlab{}.
\newblock \showarticletitle{Typed Normal Form Bisimulation}. In
  \bibinfo{booktitle}{\emph{Proc. of Computer Science Logic {CSL} 2007}}
  \emph{(\bibinfo{series}{Lecture Notes in Computer Science})},
  Vol.~\bibinfo{volume}{4646}. \bibinfo{publisher}{Springer},
  \bibinfo{pages}{283--297}.
\newblock


\bibitem[\protect\citeauthoryear{L{\'{e}}vy}{L{\'{e}}vy}{1975}]%
        {levy75}
\bibfield{author}{\bibinfo{person}{Jean{-}Jacques L{\'{e}}vy}.}
  \bibinfo{year}{1975}\natexlab{}.
\newblock \showarticletitle{An algebraic interpretation of the lambda
  beta-calculus and a labeled lambda-calculus}. In
  \bibinfo{booktitle}{\emph{Lambda-Calculus and Computer Science Theory,
  Proceedings of the Symposium Held in Rome, March 25-27, 1975}}
  \emph{(\bibinfo{series}{Lecture Notes in Computer Science})},
  Vol.~\bibinfo{volume}{37}. \bibinfo{publisher}{Springer},
  \bibinfo{pages}{147--165}.
\newblock


\bibitem[\protect\citeauthoryear{Longo}{Longo}{1983}]%
        {LONGO1983153}
\bibfield{author}{\bibinfo{person}{Giuseppe Longo}.}
  \bibinfo{year}{1983}\natexlab{}.
\newblock \showarticletitle{Set-theoretical models of lambda-calculus:
  theories, expansions, isomorphisms}.
\newblock \bibinfo{journal}{\emph{Annals of Pure and Applied Logic}}
  \bibinfo{volume}{24}, \bibinfo{number}{2} (\bibinfo{year}{1983}),
  \bibinfo{pages}{153 -- 188}.
\newblock
\showISSN{0168-0072}


\bibitem[\protect\citeauthoryear{Merro and Sangiorgi}{Merro and
  Sangiorgi}{2004}]%
        {localpi}
\bibfield{author}{\bibinfo{person}{Massimo Merro} {and} \bibinfo{person}{Davide
  Sangiorgi}.} \bibinfo{year}{2004}\natexlab{}.
\newblock \showarticletitle{On asynchrony in name-passing calculi}.
\newblock \bibinfo{journal}{\emph{Mathematical Structures in Computer Science}}
  \bibinfo{volume}{14}, \bibinfo{number}{5} (\bibinfo{year}{2004}),
  \bibinfo{pages}{715--767}.
\newblock


\bibitem[\protect\citeauthoryear{Milner}{Milner}{1990}]%
        {milner:inria-00075405}
\bibfield{author}{\bibinfo{person}{Robin Milner}.}
  \bibinfo{year}{1990}\natexlab{}.
\newblock \bibinfo{booktitle}{\emph{{Functions as processes}}}.
\newblock \bibinfo{type}{Research Report} RR-1154.
  \bibinfo{institution}{{INRIA}}.
\newblock


\bibitem[\protect\citeauthoryear{Milner}{Milner}{1992}]%
        {encodingsmilner}
\bibfield{author}{\bibinfo{person}{Robin Milner}.}
  \bibinfo{year}{1992}\natexlab{}.
\newblock \showarticletitle{Functions as Processes}.
\newblock \bibinfo{journal}{\emph{Mathematical Structures in Computer Science}}
  \bibinfo{volume}{2}, \bibinfo{number}{2} (\bibinfo{year}{1992}),
  \bibinfo{pages}{119--141}.
\newblock


\bibitem[\protect\citeauthoryear{Milner}{Milner}{1993}]%
        {milner1993polyadic}
\bibfield{author}{\bibinfo{person}{Robin Milner}.}
  \bibinfo{year}{1993}\natexlab{}.
\newblock \showarticletitle{The polyadic $\pi$-calculus: a tutorial}.
\newblock In \bibinfo{booktitle}{\emph{Logic and algebra of specification}}.
  \bibinfo{series}{NATO ASI Series (Series F: Computer \& Systems Sciences)},
  Vol.~\bibinfo{volume}{94}. \bibinfo{publisher}{Springer},
  \bibinfo{pages}{203--246}.
\newblock


\bibitem[\protect\citeauthoryear{Ong and Gianantonio}{Ong and
  Gianantonio}{2004}]%
        {DBLP:journals/tcs/OngG04}
\bibfield{author}{\bibinfo{person}{C.{-}H.~Luke Ong} {and}
  \bibinfo{person}{Pietro~Di Gianantonio}.} \bibinfo{year}{2004}\natexlab{}.
\newblock \showarticletitle{Games characterizing Levy-Longo trees}.
\newblock \bibinfo{journal}{\emph{Theor. Comput. Sci.}} \bibinfo{volume}{312},
  \bibinfo{number}{1} (\bibinfo{year}{2004}), \bibinfo{pages}{121--142}.
\newblock


\bibitem[\protect\citeauthoryear{Plotkin}{Plotkin}{1975}]%
        {DBLP:journals/tcs/Plotkin75}
\bibfield{author}{\bibinfo{person}{Gordon~D. Plotkin}.}
  \bibinfo{year}{1975}\natexlab{}.
\newblock \showarticletitle{Call-by-Name, Call-by-Value and the
  lambda-Calculus}.
\newblock \bibinfo{journal}{\emph{Theor. Comput. Sci.}} \bibinfo{volume}{1},
  \bibinfo{number}{2} (\bibinfo{year}{1975}), \bibinfo{pages}{125--159}.
\newblock


\bibitem[\protect\citeauthoryear{Rocca and Paolini}{Rocca and Paolini}{2004}]%
        {DBLP:series/txtcs/RoccaP04}
\bibfield{author}{\bibinfo{person}{Simona Ronchi~Della Rocca} {and}
  \bibinfo{person}{Luca Paolini}.} \bibinfo{year}{2004}\natexlab{}.
\newblock \bibinfo{booktitle}{\emph{The Parametric Lambda Calculus - {A}
  Metamodel for Computation}}.
\newblock \bibinfo{publisher}{Springer}.
\newblock
\showISBNx{978-3-642-05746-5}


\bibitem[\protect\citeauthoryear{Sangiorgi}{Sangiorgi}{1993a}]%
        {sangiorgiphd}
\bibfield{author}{\bibinfo{person}{Davide Sangiorgi}.}
  \bibinfo{year}{1993}\natexlab{a}.
\newblock \emph{\bibinfo{title}{Expressing mobility in process algebras :
  first-order and higher-order paradigms}}.
\newblock \bibinfo{thesistype}{Ph.D. Dissertation}. \bibinfo{school}{University
  of Edinburgh, {UK}}.
\newblock


\bibitem[\protect\citeauthoryear{Sangiorgi}{Sangiorgi}{1993b}]%
        {San93}
\bibfield{author}{\bibinfo{person}{Davide Sangiorgi}.}
  \bibinfo{year}{1993}\natexlab{b}.
\newblock \showarticletitle{An investigation into Functions as Processes}. In
  \bibinfo{booktitle}{\emph{Proc.\ of {MFPS}'93}}
  \emph{(\bibinfo{series}{Lecture Notes in Computer Science})},
  Vol.~\bibinfo{volume}{802}. \bibinfo{publisher}{Springer},
  \bibinfo{pages}{143--159}.
\newblock


\bibitem[\protect\citeauthoryear{Sangiorgi}{Sangiorgi}{1996}]%
        {internalpi}
\bibfield{author}{\bibinfo{person}{Davide Sangiorgi}.}
  \bibinfo{year}{1996}\natexlab{}.
\newblock \showarticletitle{$\pi$-Calculus, Internal Mobility, and
  Agent-Passing Calculi}.
\newblock \bibinfo{journal}{\emph{Theor. Comput. Sci.}} \bibinfo{volume}{167},
  \bibinfo{number}{1{\&}2} (\bibinfo{year}{1996}), \bibinfo{pages}{235--274}.
\newblock


\bibitem[\protect\citeauthoryear{Sangiorgi}{Sangiorgi}{2000}]%
        {cbn}
\bibfield{author}{\bibinfo{person}{Davide Sangiorgi}.}
  \bibinfo{year}{2000}\natexlab{}.
\newblock \showarticletitle{Lazy functions and mobile processes}. In
  \bibinfo{booktitle}{\emph{Proof, Language, and Interaction, Essays in Honour
  of Robin Milner}}. \bibinfo{publisher}{The {MIT} Press},
  \bibinfo{pages}{691--720}.
\newblock
\showISBNx{978-0-262-16188-6}


\bibitem[\protect\citeauthoryear{Sangiorgi}{Sangiorgi}{2006}]%
        {termination3}
\bibfield{author}{\bibinfo{person}{Davide Sangiorgi}.}
  \bibinfo{year}{2006}\natexlab{}.
\newblock \showarticletitle{Termination of processes}.
\newblock \bibinfo{journal}{\emph{Mathematical Structures in Computer Science}}
  \bibinfo{volume}{16}, \bibinfo{number}{1} (\bibinfo{year}{2006}),
  \bibinfo{pages}{1--39}.
\newblock


\bibitem[\protect\citeauthoryear{Sangiorgi and Walker}{Sangiorgi and
  Walker}{2001}]%
        {SW01a}
\bibfield{author}{\bibinfo{person}{Davide Sangiorgi} {and}
  \bibinfo{person}{David Walker}.} \bibinfo{year}{2001}\natexlab{}.
\newblock \bibinfo{booktitle}{\emph{The Pi-Calculus - a theory of mobile
  processes}}.
\newblock \bibinfo{publisher}{Cambridge University Press}.
\newblock
\showISBNx{978-0-521-78177-0}


\bibitem[\protect\citeauthoryear{Sangiorgi and Xu}{Sangiorgi and Xu}{2014}]%
        {xian}
\bibfield{author}{\bibinfo{person}{Davide Sangiorgi} {and}
  \bibinfo{person}{Xian Xu}.} \bibinfo{year}{2014}\natexlab{}.
\newblock \showarticletitle{Trees from Functions as Processes}. In
  \bibinfo{booktitle}{\emph{Proceedings of {CONCUR} 2014}}
  \emph{(\bibinfo{series}{Lecture Notes in Computer Science})},
  Vol.~\bibinfo{volume}{8704}. \bibinfo{publisher}{Springer},
  \bibinfo{pages}{78--92}.
\newblock


\bibitem[\protect\citeauthoryear{St{\o}vring and Lassen}{St{\o}vring and
  Lassen}{2009}]%
        {lassenfa}
\bibfield{author}{\bibinfo{person}{Kristian St{\o}vring} {and}
  \bibinfo{person}{S{\o}ren~B. Lassen}.} \bibinfo{year}{2009}\natexlab{}.
\newblock \showarticletitle{A Complete, Co-inductive Syntactic Theory of
  Sequential Control and State}. In \bibinfo{booktitle}{\emph{Semantics and
  Algebraic Specification, Essays Dedicated to Peter D. Mosses on the Occasion
  of His 60th Birthday}} \emph{(\bibinfo{series}{Lecture Notes in Computer
  Science})}, Vol.~\bibinfo{volume}{5700}. \bibinfo{publisher}{Springer},
  \bibinfo{pages}{329--375}.
\newblock


\bibitem[\protect\citeauthoryear{Toninho and Yoshida}{Toninho and
  Yoshida}{2018}]%
        {toninho:yoshida:esop18}
\bibfield{author}{\bibinfo{person}{Bernardo Toninho} {and}
  \bibinfo{person}{Nobuko Yoshida}.} \bibinfo{year}{2018}\natexlab{}.
\newblock \showarticletitle{On Polymorphic Sessions and Functions - {A} Tale of
  Two (Fully Abstract) Encodings}. In \bibinfo{booktitle}{\emph{Proc. of {ESOP}
  2018}} \emph{(\bibinfo{series}{Lecture Notes in Computer Science})},
  Vol.~\bibinfo{volume}{10801}. \bibinfo{publisher}{Springer},
  \bibinfo{pages}{827--855}.
\newblock


\bibitem[\protect\citeauthoryear{Yoshida, Berger, and Honda}{Yoshida
  et~al\mbox{.}}{2004}]%
        {termination1}
\bibfield{author}{\bibinfo{person}{Nobuko Yoshida}, \bibinfo{person}{Martin
  Berger}, {and} \bibinfo{person}{Kohei Honda}.}
  \bibinfo{year}{2004}\natexlab{}.
\newblock \showarticletitle{Strong normalisation in the pi -calculus}.
\newblock \bibinfo{journal}{\emph{Inf. Comput.}} \bibinfo{volume}{191},
  \bibinfo{number}{2} (\bibinfo{year}{2004}), \bibinfo{pages}{145--202}.
\newblock


\end{thebibliography}

\end{document}

\end{document}